\documentclass{article}
\usepackage{tikz,hyperref,amsmath,amssymb,ntheorem}

\newtheorem*{maintheorem}{Main Theorem}
\newtheorem{theorem}{Theorem}[section]
\newtheorem{lemma}[theorem]{Lemma}
\newtheorem{proposition}[theorem]{Proposition}
\newtheorem{corollary}[theorem]{Corollary}
{\theorembodyfont{\rmfamily}
  \newtheorem{example}[theorem]{Example}
   
  \newtheorem{definition}[theorem]{Definition}}
\newenvironment{proof}{\noindent\textit{Proof.}}{\QED\vskip\theorempostskipamount} 
\def\petitcarre{\vrule height4pt width 4pt depth0pt}
\def\QED{\relax\ifmmode\eqno{\hbox{\petitcarre}}\else{%
  \unskip\nobreak\hfil\penalty50\hskip2em\hbox{}\nobreak\hfil
  \petitcarre
  \parfillskip=0pt \finalhyphendemerits=0\par\smallskip}
  \fi}
\numberwithin{equation}{section}
\numberwithin{figure}{section}

\newcommand\cL{\mathcal{L}}
\newcommand{\Z}{\mathbb{Z}}
\DeclareMathOperator{\Card}{Card}

\definecolor{lime}{HTML}{A6CE39}
\DeclareRobustCommand{\orcidicon}{%
	\begin{tikzpicture}
	\draw[lime, fill=lime] (0,0)
	circle [radius=0.16]
	node[white] {{\fontfamily{qag}\selectfont \tiny ID}};
	\draw[white, fill=white] (-0.0625,0.095)
	circle [radius=0.007];
	\end{tikzpicture}
	\hspace{-2mm}
}
\foreach \x in {A, ..., Z}{%
 \expandafter\xdef\csname orcid\x\endcsname{\noexpand%
 \href{https://orcid.org/\csname orcidauthor\x\endcsname}{\noexpand\orcidicon}}
}

\title{Recognizability in $S$-adic shifts}

\author{Marie-Pierre B\'eal\orcidA{}, Dominique Perrin, \\ Antonio Restivo, and Wolfgang Steiner}

\begin{document}
\maketitle
\begin{abstract}
We investigate questions related to the notion of recognizability of sequences of morphisms, a generalization of Moss\'e's Theorem.
We consider the most general class of morphisms including ones with erasable letters.
The main result states that a sequence of morphisms with finite alphabet rank is eventually recognizable for aperiodic points, improving and simplifying a result of Berth\'e et al.\ (2019). 
This also provides a new simple proof for the recognizability of a single morphism on its shift space.
The main ingredient of the proof is elementary morphisms.
\end{abstract}

\section{Introduction}
Given a bi-infinite sequence in $A^{\Z}$ and a morphism (also called a substitution) $\sigma\colon A^*\to A^*$, recognizability is a form of injectivity of~$\sigma$ that allows one to uniquely desubstitute $y$ to another sequence~$x$, i.e., to express $y$ as a concatenation of substitution words dictated by the letters in~$x$. 
The sequences $y$ and~$x$ are traditionally required to be in the shift space $X(\sigma)$, which is the set of bi-infinite sequences (also called points) whose finite factors are factors of $\sigma^n(a)$ for some integer $n$ and some letter $a$ in~$A$. 

By Moss\'e's Theorem \cite{Mosse1992,Mosse1996}, every aperiodic
primitive morphism~$\sigma$ is recognizable in the shift $X(\sigma)$; see the precise definitions in Section~\ref{sectionRecognizability}. 
This surprising result was initially formulated (in an incorrect 
way) by \cite{Martin1973}; see \cite{Kyriakoglou2019} on the genesis
of the theorem and its possible variants. It was further generalized
by Bezuglyi, Kwiatkowski, and Medynets~\cite{BezuglyiKwiatkowskiMedynets2009},
who proved that every aperiodic non-erasing morphism $\sigma$
is recognizable in $X(\sigma)$. Next, it was proved 
by Berth\'e, Thuswaldner, Yassawi, and the fourth author~\cite{BertheSteinerThuswaldnerYassawi2019}
that every non-erasing morphism~$\sigma$ is recognizable in $X(\sigma)$
for aperiodic points, and the first three authors proved in~\cite{BealPerrinRestivo2023}
that every morphism $\sigma$ is recognizable in $X(\sigma)$ for aperiodic points.

There is a strong link between recognizability and automata theory due
to a translation of the property of recognizability in terms of finite
monoids. For instance, there is a quadratic-time algorithm to check whether an
injective morphism is recognizable in the full shift for aperiodic
points~\cite{BealPerrinRestivo2023}.

In this paper, we investigate recognizability in the
context of sequences of morphisms
$\sigma = (\sigma_n\colon A_{n+1}^* \to A_n^*)_{n \geq 0}$. 
Such a sequence defines an $S$-adic shift, generated by
iterations of the form
$\sigma_0 \circ \sigma_1 \circ \cdots \circ \sigma_n$; see
Section~\ref{sectionRecognizabilitySadic} for the definition.
In fact, $\sigma$ defines a sequence $X^{(n)}(\sigma)$ of shift spaces and here, by recognizability of $\sigma_n$, we mean that any
sequence in $X^{(n)}(\sigma)$ can be desubstituted in at most one (and usually exactly one) way as sequences in $X^{(n+1)}(\sigma)$ using~$\sigma_n$.
We distinguish between
recognizability of $\sigma$, where each $\sigma$ is recognizable,
and eventual recognizability, where
all but finitely many morphisms $\sigma_n$ are recognizable.
We consider the most general class of sequences of morphisms,
including ones with morphisms with erasable letters.

Recognizability of $S$-adic shifts has been studied in~\cite{BertheSteinerThuswaldnerYassawi2019}, where it is proved that
a morphism $\sigma$ is eventually recognizable for aperiodic points in its
shift spaces
under some mild conditions:
the morphisms $\sigma_n$ are non-erasing, the sequence $\Card(A_n)$ contains a bounded subsequence, and $\sigma$ is everywhere growing (or the
points in each $X^{(n)}(\sigma)$ generate a bounded number of
different languages). This means that, for large enough $n$,
every aperiodic point  in $X^{(n)}(\sigma)$
has a unique centered $\sigma_n$-represen\-tation as a shift
of the image by $\sigma_n$ of some $x\in X^{(n+1)}(\sigma)$.
This recognizability property implies a natural
representation of an $S$-adic shift as a Bratteli–Vershik system~\cite{BertheSteinerThuswaldnerYassawi2019}.
A weaker notion, called quasi-recognizability, is studied in \cite{BustosGajardoManiboYassawi2023}.

The main result of this paper is a generalization to possibly erasing
sequences of morphisms of the result of
\cite{BertheSteinerThuswaldnerYassawi2019} concerning
the recognizability of a sequence of morphisms in its shift spaces
for aperiodic points. 

When a morphism $\sigma_n$ erases a letter, it is possible that a sequence in $X^{(n)}(\sigma)$ cannot be desubstituted as a sequence in $X^{(n+1)}(\sigma)$ using~$\sigma_n$.
We are therefore not only concerned with recognizability but also with representability, which means that $X^{(n)}(\sigma)$ is the shift closure of the image of $X^{(n+1)}(\sigma)$ by~$\sigma_n$; see Section~\ref{sec:representable-s-adic} for details.

We prove the following result, where the alphabet rank of a sequence of morphisms $(\sigma_n\colon A_{n+1}^* \to A_n^*)_{n \geq 0}$ is $\liminf_{n\to\infty} \Card(A_n)$; a more precise statement is given in Section~\ref{sec:levels-recogn-repr}.

\begin{maintheorem}
Any sequence of morphisms with finite alphabet rank is eventually recognizable for aperiodic points and eventually representable.
\end{maintheorem}

Our proof is much simpler than that of \cite{BertheSteinerThuswaldnerYassawi2019}, and we we do not require that $\Card(\{\cL_x \mid x \in X^{(n)}(\sigma)\})$
is bounded, where $\cL_x$ is the set of factors of a point~$x$.
Moreover, our proof gives a bound equal to the alphabet rank minus $2$
on the number of levels at which $\sigma$ is not recognizable for aperiodic points, improving the bound of order $K(K + L \log K)$ obtained in \cite{BertheSteinerThuswaldnerYassawi2019}
for alphabets of size at most $K$ and $\Card(\{\cL_x \mid x \in
X^{(n)}(\sigma)\}) \leq L$. We also show that this bound is tight.

Our result allows one also to get a new simpler proof of the recognizability of a (possibly erasable) morphism~$\sigma$ on $X(\sigma)$ for aperiodic points obtained in~\cite{BealPerrinRestivo2023}.

As in \cite{BealPerrinRestivo2023}, our
proof relies on the notion of elementary morphism,
due to Ehrenfeucht and Rozenberg \cite{EhrenfeuchtRozenberg1978}.
By a result of Karhum\"aki, Ma\v{n}uch and Plan\-dowski~\cite{KarhumakiManuchPlandowski2003}, every
elementary morphism is recognizable for aperiodic points; see also~\cite{BertheSteinerThuswaldnerYassawi2019}.
We use this result to prove eventual recognizability.

The paper is organized as follows. After an introductory section
on basic notions of symbolic dynamics, we formulate
the precise definition of a morphism recognizable
on a shift space and prove some elementary properties
of recognizable morphisms. In Section \ref{sectionElementary},
we introduce elementary morphisms and recall that every
elementary morphism is recognizable for
aperiodic points (Proposition~\ref{theoremIndecomposable}).
The main results are proved in
Section~\ref{sectionRecognizabilitySadic}.

\section{Symbolic dynamics}
We briefly recall some basic definitions of symbolic dynamics. For
a more complete presentation, see~\cite{LindMarcus2021}
or the recent~\cite{DurandPerrin2021}.

\subsection{Words}
Let $A$ be a finite alphabet. We let $A^*$ denote the free
monoid on~$A$, i.e., the set of finite words over the alphabet~$A$.
The empty word is denoted by~$\varepsilon$.
We let $|u|$ denote the \emph{length} of the word $u$.

A word $s\in A^*$ is a \emph{factor} of $w\in A^*$ if $w=rst$; the word $r$ is called a
\emph{prefix} of $w$, and it is \emph{proper} if $r\ne w$.

\subsection{Shift spaces}
We consider the set $A^\Z$ of two-sided infinite sequences (also
called points) on~$A$.
For $x=(x_n)_{n\in\Z}$, and $i\le j$, we let $x_{[i,j)}$ denote the word $x_ix_{i+1}\cdots x_{j-1}$, where $x_{[i,i)}$ is the empty word; the word $x_{[i,j)}$ is called a \emph{factor} of~$x$.

The set $A^\Z$ is a compact metric space for the distance defined for $x\ne y$
by $d(x,y)=2^{-\min\{|n|\,\mid\, n\in\Z, x_n\ne y_n\}}$.
The \emph{shift transformation} $T \colon A^\Z\to A^\Z$ is defined
by $T((x_n)_{n\in\Z}) = (x_{n+1})_{n\in\Z}$.
A~\emph{shift space} $X$ on a finite alphabet $A$ is a closed
and shift-invariant subset of $A^\Z$.

A~point $x\in A^\Z$ is \emph{periodic} if there is an $n\ge 1$
such that $T^n(x)=x$. Otherwise, it is \emph{aperiodic}.
A periodic point has the form $w^\infty=\cdots ww\cdot ww\cdots$ (the letter of index $0$ of $w^\infty$ is the first letter of $w$).

\subsection{Morphisms}
A morphism $\sigma\colon A^*\to B^*$ is a monoid morphism from $A^*$ to $B^*$.
It is \emph{erasing} if $\sigma(a)$ is the empty word for some $a \in A$, \emph{non-erasing} otherwise.  

The morphism $\sigma$ is extended to a map from $A^\Z$ to $B^\Z$ by
\[
\sigma(\cdots x_{-2} x_{-1} \cdot x_0 x_1 \cdots) = \cdots \sigma(x_{-2}) \sigma(x_{-1}) \cdot \sigma(x_0) \sigma(x_1) \cdots,
\]
i.e., $\sigma((x_n)_{n\in\Z}) = (y_n)_{n\in\Z}$ with $y_{[|\sigma(x_{[0,n)})|, |\sigma(x_{[0,n+1)})|)} = \sigma(x_n)$ for all $n \ge 0$ and $y_{[|\sigma(x_{[n,0)})|, |\sigma(x_{[n+1,0)})|)} = \sigma(x_n)$ for all $n < 0$; this map is defined only for points in $A^\Z$ containing infinitely many letters on the left and infinitely many letters on the right that are not erased.

Let $\sigma\colon A^*\to A^*$ be a morphism from $A^*$ to itself.
For $n \geq 0$, we let $\sigma^n$ denote the morphism obtained with $n$ iterations of $\sigma$.
The \emph{language of~$\sigma$}, denoted~$\cL(\sigma)$,
is the set of factors of the words $\sigma^n(a)$ for
some $n\ge 0$ and $a\in A$.
The \emph{shift defined by~$\sigma$}, denoted by $X(\sigma)$,
is the set of sequences with all their factors in~$\cL(\sigma)$.
The morphism~$\sigma$ is \emph{primitive} if there exists $n \ge 1$ such that
the word $\sigma^n(a)$ contains the letter~$b$ for all $a, b \in A$. 

 \subsection{$S$-adic shifts} \label{sectionSadic}
Let $\sigma = (\sigma_n)_{n \geq 0}$ be a sequence of morphisms
$\sigma_n\colon A_{n+1}^*\to A_n^*$, where $A_n$ are finite alphabets:
\begin{displaymath}
A_0^* \xleftarrow{\sigma_0} A_1^* \xleftarrow{\sigma_1} A_2^*  \xleftarrow{\sigma_2} \cdots \xleftarrow{\sigma_{n-1}} A_n^* \xleftarrow{\sigma_n} A_{n+1}^*  \xleftarrow{\sigma_{n+1}} \cdots.
\end{displaymath}
For $0 \leq n \le  m$, we define the morphism $\sigma_{[n,m)} \colon A_m^* \to A_n^*$ by
\begin{displaymath}
\sigma_{[n, m)} = \sigma_n \circ \sigma_{n+1} \circ \cdots \circ \sigma_{m-1},
\end{displaymath}
where $\sigma_{[n,n)}$ is the identity. 
For $n \geq 0$, the language of $\cL^{(n)}(\sigma)$ is the subset of $A_n^*$ of factors of the words $\sigma_{[n, m)}(a)$, $a\in A_m$, $m \ge n$, and the shift $X^{(n)}(\sigma)$ is the set of sequences with all their factors in $\cL^{(n)}(\sigma)$.
The \emph{$S$-adic shift} defined by~$\sigma$ is~$X^{(0)}(\sigma)$.

A~sequence of morphisms $\sigma = (\sigma_n)_{n \geq 0}$ is
\emph{non-erasing} if all $\sigma_n$ are non-erasing.
It is primitive if for each $n \ge 0$ there exists $m > n$ such that
the word $\sigma_{[n,m)}(a)$ contains the letter~$b$ for all $a \in A_m$, $b \in A_n$. 

\section{Recognizable morphisms}\label{sectionRecognizability}
Let $\sigma\colon A^*\to B^*$ be a morphism. A \emph{$\sigma$-representation}
of $y\in B^\Z$ is a pair $(x,k)$ of a sequence $x\in A^\Z$
and an integer $k$  such that
\begin{equation}
y = T^k(\sigma(x)), \label{eqsigmaRep}
\end{equation}
where $T$ is the shift transformation.
The $\sigma$-representation $(x,k)$ is \emph{centered}
if $0\le k<|\sigma(x_0)|$.
In particular, a centered $\sigma$-representation $(x,k)$ satisfies $\sigma(x_0)\ne\varepsilon$. 
We say that the $\sigma$-representation $(x,k)$ is \emph{in $X$} if $x \in X$.

Note that, if $y$ has a $\sigma$-representation $(x,k)$, then it has a
centered $\sigma$-representation $(x',k')$ with $x'$ a shift of~$x$.

\begin{definition}
Let $X$ be a shift space on $A$.
A morphism $\sigma\colon A^*\to B^*$ is \emph{recognizable in $X$}
(respectively recognizable in $X$
for aperiodic points) if every point in $B^\Z$ (respectively
every aperiodic point in $B^\Z$) has at most one centered
$\sigma$-representation in~$X$.
A morphism $\sigma\colon A^*\to B^*$ is \emph{fully recognizable}
(respectively fully recognizable
for aperiodic points) if it is recognizable in~$A^\Z$ (respectively
recognizable in $A^\Z$
for aperiodic points).
\end{definition}

Note that an equivalent definition of recognizability in $X$ is that, for every $x,x'\in X$ and $0\le k<|\, |\sigma(x'_0)|-|\sigma(x_0)|\, |$ such that
$\sigma(x) = T^k(\sigma(x'))$, one has $k=0$ and $x=x'$.

\begin{example}
  The Fibonacci morphism $\sigma\colon a\mapsto ab,b\mapsto a$
  is fully recognizable.
\end{example}
\begin{example}
  The Thue-Morse morphism $\sigma\colon a\mapsto ab,b\mapsto ba$ is
  not fully recognizable since $(ab)^\infty$
  can be obtained as $\sigma(a^\infty)$ and as $T(\sigma(b^\infty))$.
  However, it is fully recognizable for aperiodic points since
  any sequence containing $aa$ or $bb$ has at most one factorization
  in $\{ab,ba\}$.
\end{example}
\begin{example}\label{example3.4}
  The morphism $\sigma\colon a\to aa,b\mapsto ab,c\mapsto ba$
  is not fully recognizable for aperiodic points. Indeed,
  every sequence without occurrence of $bb$ has
  two factorizations in words of $\{aa,ab,ba\}$.
  \end{example}

By \cite{BertheSteinerThuswaldnerYassawi2019,BealPerrinRestivo2023}, the family of morphisms recognizable for aperiodic points is closed under composition.

\section{Elementary morphisms}\label{sectionElementary}
\begin{definition}
A morphism $\sigma\colon A^*\to C^*$ is \emph{elementary}
if
for every alphabet $B$ and
every pair of morphisms $\alpha\colon B^*\to C^*$ and $\beta\colon A^*\to B^*$
such that $\sigma=\alpha\circ\beta$, one has
$\Card(B)\ge \Card(A)$.
\end{definition}

If $\sigma\colon A^*\to C^*$ is elementary,  one has in
particular $\Card(C)\ge\Card(A)$ and moreover $\sigma$ is non-erasing.

\begin{example}\label{exampleTMelementary}
The Thue-Morse morphism $\sigma\colon a\mapsto ab,b\mapsto ba$ is elementary. 
\end{example}

The notion of elementary morphism appears for the first time
in \cite{EhrenfeuchtRozenberg1978}.
The following result is from~\cite{KarhumakiManuchPlandowski2003}.
It also appears in \cite{BertheSteinerThuswaldnerYassawi2019}
with the stronger hypothesis that
$\sigma\colon A^*\to B^*$ is such that the incidence matrix of $\sigma$ has
rank $\Card(A)$. An independent proof is given in \cite{BealPerrinRestivo2023}.

\begin{proposition} \label{theoremIndecomposable}
Any elementary morphism is fully recognizable for aperiodic points. 
\end{proposition}

\section{Representable $S$-adic shifts} \label{sec:representable-s-adic}
Contrary to sequences of non-erasing morphisms $\sigma$, a point in $X^{(n)}(\sigma)$ need not have a $\sigma_n$-representation in $X^{(n+1)}(\sigma)$ when $\sigma_n$ is erasing, as the following example shows. 

\begin{example} \label{ex:nonrepresentable}
Let the sequence of morphisms $\sigma$ be defined by
\begin{align*}
\sigma_0 & \colon a \mapsto a, b \mapsto \varepsilon, \\
\sigma_1 & \colon a \mapsto a, b \mapsto bb, c \mapsto ab, \\
\sigma_n & \colon a \mapsto a, b \mapsto bb, c \mapsto cab, \quad \text{for all } n \geq 2.
\end{align*}
Since $\sigma_{[1,n)}(a) = a$, $\sigma_{[1,n)}(b) = b^{2^{n-1}}$, $\sigma_{[1,n)}(c) = abab^2\cdots ab^{2^{n-2}}$ for all $n \ge 2$, we have $X^{(0)}(\sigma) = \{a^\infty\}$ and $X^{(1)}(\sigma)$ consists of the points in $\{a,b\}^\mathbb{Z}$ containing at most one~$a$, hence $a^\infty$ has no $\sigma_0$-representation in $X^{(1)}(\sigma)$.
\end{example}

We say that a sequence of morphisms
$\sigma$ is \emph{representable at level~$n$} 
if  every point in $X^{(n)}(\sigma)$ has at least one
$\sigma_n$-representation in $X^{(n+1)}(\sigma)$.
It is \emph{representable} if it is representable at each level.
We say that a sequence of morphisms is \emph{eventually representable} if there is an integer $M$ such that it is representable at each level at least equal to $M$.

Note that $X^{(n)}(\sigma)$ is the shift-closure of $\sigma_n(X^{(n+1)}(\sigma))$ if and only if $\sigma$ is representable at level~$n$.
 
The following lemma is proved in \cite[Lemma~4.2]{BertheSteinerThuswaldnerYassawi2019} for sequences of non-erasing morphisms.
We recall its proof to make clear where the non-erasing property is used. 

\begin{lemma} \label{l:representable}
Let $\sigma = (\sigma_n)_{n \geq 0}$ with $\sigma_n\colon A_{n+1}^* \to A_n^*$ be a sequence of morphisms.
If $\sigma_{[n,m)}$ is non-erasing, $0 \le n < m$, then every point in $X^{(n)}(\sigma)$ has at least one $\sigma_{[n,m)}$-representation in~$X^{(m)}(\sigma)$.
In particular, if $\sigma_n$ is non-erasing, then $\sigma$ is representable at level~$n$. 
\end{lemma}

\begin{proof}
Let $y \in X^{(n)}(\sigma)$.
Then each word $y_{[-\ell,\ell)}$ is a factor of $\sigma_{[n,N)}(a)$ for some $a \in A_N$, $N \ge m$, hence $y_{[-\ell+i,\ell-j)} = \sigma_{[n,m)}(w)$ for some $w \in \cL^{(m)}(\sigma)$, $0 \le i,j < \max_{a\in A_m} |\sigma_{[n,m)}(a)|$.
Since $|w| \to \infty$ as $\ell \to \infty$, a Cantor diagonal argument gives a word $x \in X^{(m)}(\sigma)$ and $0 \le k < |\sigma_{[n,m)}(x_0)|$ such that $\sigma_{[n,m)}(x_{[-\ell,\ell)}) = y_{[|\sigma_{[n,m)}(x_{[-\ell,0)})|-k,|\sigma_{[n,m)}(x_{[0,\ell)})|-k)}$ for all $\ell \ge 1$. 
Since $\sigma_{[n,m)}$ is non-erasing, $(k,x)$ is a $\sigma_{[n,m)}$-representation of~$y$.
\end{proof}

\begin{lemma} \label{l:notrepresentableerasing}
If $\sigma$ is not representable at level~$n$, then $\sigma_{[n,m)}$ is erasing for all $m > n$. 
\end{lemma}

\begin{proof}
If $\sigma_{[n,m)}$ is non-erasing, then, by Lemma~\ref{l:representable}, each $y \in X^{(n)}(\sigma)$ has $\sigma_{[n,m)}$-representation in $X^{(m)}(\sigma)$, thus it also has a $\sigma_n$-representation in $X^{(n+1)}(\sigma)$, i.e., $\sigma$ is representable at level~$n$.
\end{proof}

\section{Recognizable $S$-adic shifts} \label{sectionRecognizabilitySadic}
A~sequence of morphisms $\sigma= (\sigma_n)_{n \geq 0}$ with $\sigma_n\colon A_{n+1}^*\to A_n^*$ is \emph{recognizable at level $n$} (respectively \emph{recognizable at level $n$ for aperiodic points}) if $\sigma_n$ is recognizable
(respectively recognizable for aperiodic points) in
$X^{(n+1)}(\sigma)$.
We say that $\sigma$ is \emph{recognizable} (respectively
\emph{recognizable for aperiodic points})  if it is recognizable (respectively recognizable for aperiodic points) at each nonnegative level~$n$, and $\sigma$ is \emph{eventually recognizable} (respectively
\emph{eventually recognizable for aperiodic points}) if there is a nonnegative
integer $M$ such that $\sigma$ is recognizable
(respectively recognizable for aperiodic points) at level~$n$ for each $n \geq M$.

We show that non-recognizability at level~$n$ and representability between levels $n{+}1$ and~$m$ implies non-recognizability between levels $n$ and~$m$. 

\begin{lemma} \label{l:recognizablerepresentable}
If $\sigma$ is not recognizable at level~$n$ and each point in $X^{(n+1)}(\sigma)$ has a $\sigma_{[n+1,m)}$-representation in $X^{(m)}(\sigma)$, then $\sigma_{[n,m)}$ is not recognizable in~$X^{(m)}(\sigma)$. 
The same statement holds for recognizability for aperiodic points.
\end{lemma}

\begin{proof}
This is proved in \cite[Lemma~3.5]{BertheSteinerThuswaldnerYassawi2019} for non-erasing morphisms; we recall the proof.
If $\sigma$ is not recognizable at level~$n$, then there exists $z \in X^{(n)}(\sigma)$ with two centered $\sigma_n$-representations $(y,\ell) \ne (y,\ell')$ in~$X^{(n+1)}(\sigma)$. 
Let $(x,k)$ and $(x',k')$ be centered $\sigma_{[n+1,m)}$-representations in~$X^{(m)}(\sigma)$ of $y$ and~$y'$, respectively.
Then $(x,|\sigma_n(y_{[-k,0)})|{+}\ell)$ and $(x',|\sigma_n(y'_{[-k',0)})|{+}\ell')$ are centered $\sigma_{[n,m)}$-representations of~$z$.
To see that the two representations are different, note that $y_{[-k,0)}$, $z_{[-\ell,0)}$, and $\sigma_n(y_{[-k,0)})z_{[-\ell,0)}$ are proper prefixes of $\sigma_{[n+1,m)}(x_0)$, $\sigma_n(y_0)$, and $\sigma_{[n,m)}(x_0)$ respectively.
Since each proper prefix of $\sigma_{[n,m)}(x_0)$ has a unique
decomposition as $\sigma_n(u) v$ with $u \in A_{n+1}^*$, $v \in A_n^*$, such that $ua$ is a prefix of $\sigma_{[n+1,m)}(x_0)$ and $v$ is a proper prefix of $\sigma_n(a)$ for some $a \in A_{n+1}$, $(x,|\sigma_n(y_{[-k,0)})|{+}\ell) = (x',|\sigma_n(y'_{[-k',0)})|{+}\ell')$ would imply that $k = k'$ and $\ell = \ell'$, thus $y = y'$, contradicting that $(y,\ell) \ne (y,\ell')$.
Therefore, $\sigma_{[n,m)}$ is not recognizable on~$X^{(m)}(\sigma)$.

Taking aperiodic points $y,y'$ proves the statement for aperiodic points.
\end{proof}

\section{Levels of recognizability and representability} \label{sec:levels-recogn-repr}

We can now state and prove our main results, which give bounds for the number of levels where a sequence of morphisms can be non-recognizable for aperiodic points or non-representable, in terms of the size of the alphabets. 

\begin{proposition} \label{p:main}
Let $\sigma= (\sigma_n)_{n \geq 0}$ with $\sigma_n \colon A_{n+1}^* \to A_n^*$ be a sequence of morphisms.
Let $m > n_1 > n_2 > \cdots > n_K \ge 0$, $K \ge 0$, be such that, for each $1 \le k \le K$, $\sigma$~is not recognizable at level~$n_k$ for aperiodic points or $\sigma$~is not representable at level~$n_k$.
Then we have $K < \Card(A_m)$.
Moreover, $K = \Card(A_m) {-} 1 \ge 1$ implies that $X^{(n_K)}(\sigma)$ has no aperiodic points.
\end{proposition}

\begin{proof}
Let $m > n_1 > n_2 > \cdots > n_K \ge 0$ be as in the statement of the proposition. 
Since the proposition is trivial for $K = 0$, we assume that $K \ge 1$. 

We define $\alpha_0 \colon A_m^* \to A_m^*$ as the identity morphism and set $n_0 = m$, $B_0 = A_m$. 
For each $1 \leq k < K$,  we show inductively that the morphism
$\sigma_{[n_k,n_{k-1})} \circ \alpha_{k-1}$ is not elementary and
admits therefore a decomposition $\sigma_{[n_k,n_{k-1})} \circ
\alpha_{k-1} = \alpha_k \circ \beta_k$ (see Figure \ref{figure.proof})
with morphisms $\alpha_k \colon B_k^* \to A_{n_k}^*$, $\beta_k \colon B_{k-1}^* \to B_k^*$, for some alphabet $B_k$ satisfying $\Card(B_k) < \Card(B_{k-1})$. 

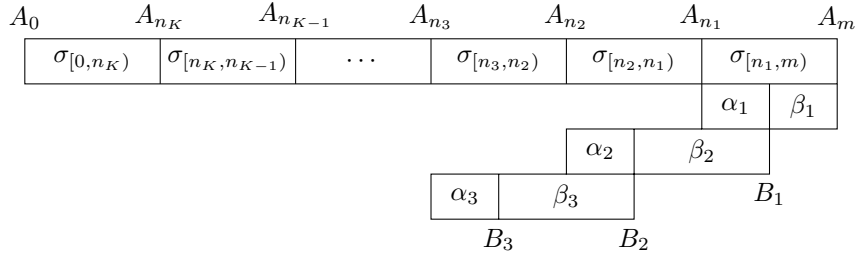
\begin{figure}[ht]
\centering
\tikzset{node/.style={draw,minimum size=.6cm,inner sep=0.4pt}}
\tikzset{title/.style={minimum size=0.5cm,inner sep=0pt}}
\begin{tikzpicture}
\draw(0,2.4)node[above]{$A_0$}--(10.8,2.4)node[above]{$A_m$}--(10.8,1.8)--(0,1.8)--cycle;
\draw(1.8,1.8)--(1.8,2.4)node[above]{$A_{n_K}$};
\draw(3.6,1.8)--(3.6,2.4)node[above]{$A_{n_{K-1}}$};
\draw(5.4,1.8)--(5.4,2.4)node[above]{$A_{n_3}$};
\draw(7.2,1.8)--(7.2,2.4)node[above]{$A_{n_2}$};
\draw(9,1.8)--(9,2.4)node[above]{$A_{n_1}$};
\node at (.9,2.1){$\sigma_{[0,n_K)}$};
\node at (2.7,2.1){$\sigma_{[n_K,n_{K-1})}$};
\node at (4.5,2.1){$\cdots$};
\node at (6.3,2.1){$\sigma_{[n_3,n_2)}$};
\node at (8.1,2.1){$\sigma_{[n_2,n_1)}$};
\node at (9.9,2.1){$\sigma_{[n_1,m)}$};
\draw(10.8,1.8)--(10.8,1.2)--(9,1.2)--(9,1.8)
(9.9,1.8)--(9.9,.6)node[below]{$B_1$}--(7.2,.6)--(7.2,1.2)--(9,1.2)
(8.1,1.2)--(8.1,0)node[below]{$B_2$}--(5.4,0)--(5.4,.6)--(7.2,.6) (6.3,.6)--(6.3,0)node[below]{$B_3$};
\node at (10.35,1.5){$\beta_1$};
\node at (9.45,1.5){$\alpha_1$};
\node at (9,.9){$\beta_2$};
\node at (7.65,.9){$\alpha_2$};
\node at (7.2,.3){$\beta_3$};
\node at (5.85,.3){$\alpha_3$};
\end{tikzpicture}
\caption{Alphabets and morphisms in the proof of Proposition~\ref{p:main}.}
 \label{figure.proof}
\end{figure}

Indeed, consider the sequence of morphisms 
\begin{displaymath}
\sigma' = (\sigma_0,\sigma_1,\dots,\sigma_{n_k},\dots,\sigma_{n_{k-1}-1},\alpha_{k-1},\beta_{k-1},\beta_{k-2},\dots,\beta_1,\sigma_m,\sigma_{m+1},\dots).
\end{displaymath}
Since $\alpha_{k-1} \circ \beta_{k-1} \circ \beta_{k-2} \circ \cdots \circ \beta_1 = \sigma_{[n_{k-1},m)}$, we have for all $0 \le h \le n_{k-1}$ that the languages $\cL^{(h)}(\sigma)$ and $\cL^{(h)}(\sigma')$ differ only by a finite set and hence $X^{(h)}(\sigma) = X^{(h)}(\sigma')$. 
If $\sigma$ is not representable at level~$n_k$, then $\sigma'$ is also not representable at level~$n_k$ and, by Lemma~\ref{l:notrepresentableerasing}, $\sigma_{[n_k,n_{k-1})} \circ \alpha_{k-1}$ is erasing, thus non-elementary.
If $\sigma$ is not recognizable at level~$n_k$ for aperiodic points
and each point in $X^{(n_k+1)}(\sigma)$ has a
$\sigma_{[n_k+1,n_{k-1})} \circ \alpha_{k-1}$-representation in
$X^{(n_{k-1}+1)}(\sigma')$, then $\sigma_{[n_k,n_{k-1})} \circ
\alpha_{k-1}$ is not recognizable for aperiodic points in $X^{(n_{k-1}+1)}(\sigma')$ by Lemma~\ref{l:recognizablerepresentable}, thus it is non-elementary. 
Finally, if there exists a point in $X^{(n_k+1)}(\sigma)$ without $\sigma_{[n_k+1,n_{k-1})} \circ \alpha_{k-1}$-representation in $X^{(n_{k-1}+1)}(\sigma')$, then $\sigma_{[n_k+1,n_{k-1})} \circ \alpha_{k-1}$ is erasing by Lemma~\ref{l:notrepresentableerasing}, hence $\sigma_{[n_k,n_{k-1})} \circ \alpha_{k-1}$ is erasing, thus non-elementary.

We get that $\Card(A_m) > \Card(B_1) > \Card(B_2) > \cdots > \Card(B_K) \ge 1$, thus $K < \Card(A_m)$.
If $K = \Card(A_m){-}1$, then $\Card(B_K) = 1$, hence~$X^{(n_K)}(\sigma)$ consists of a single periodic orbit, thus $\sigma$ is recognizable at level~$n_K$ for aperiodic points (and thus not representable at level~$n_K$ by the assumption on~$n_K$).
\end{proof}

Let $\sigma= (\sigma_n)_{n \geq 0}$ with $\sigma_n \colon A_{n+1}^*
\to A_n^*$ be a sequence of morphisms. The \emph{alphabet rank} of
$\sigma$ is $\liminf_{n\to\infty}
\Card(A_n)$.

\begin{theorem} \label{theoremBSTY}
Let $\sigma= (\sigma_n)_{n \geq 0}$ with $\sigma_n \colon A_{n+1}^*
\to A_n^*$ be a sequence of morphisms with finite alphabet rank.
Then $\sigma$ is eventually recognizable for aperiodic points and eventually representable. 

Moreover, if $r$ is the alphabet rank, the number of levels at which $\sigma$ is not recognizable for aperiodic points is bounded by $r{-}2$ and the number of levels at which $\sigma$ is not representable is bounded by $r{-}1$.
\end{theorem}

\begin{proof}
Suppose that $\sigma$ is not representable at $K = r < \infty$ levels, then applying Proposition~\ref{p:main} for some $m$ which is larger than these levels and satisfies $\Card(A_m) = K$ gives a contradiction.
Similarly, we cannot have $K = r{-}1$ levels where $\sigma$ is not recognizable for aperiodic points by Proposition~\ref{p:main} because the level $n_K$ in Proposition~\ref{p:main} can only be non-representable for $K = r{-}1$.
\end{proof}

Note that the condition of finite alphabet rank holds in particular when the sizes of the alphabets are bounded.

In the particular case of a constant sequence of morphisms $\sigma' =
(\sigma,\sigma,\dots)$, all shifts $X^{(n)}(\sigma')$ are equal to the
shift space $X(\sigma)$ of the morphism~$\sigma$, and
non-recognizability of $\sigma$ in $X(\sigma)$ for aperiodic points means that $\sigma'$ is non-recognizable at aperiodic points at all levels. 
Since this is not possible by Theorem~\ref{theoremBSTY}, this proves the main result of
\cite{BealPerrinRestivo2023}, in a way that is simpler than all the previous proofs of recognizability. 

\begin{corollary} \label{theoremBSTY4}
Any morphism $\sigma$ is recognizable for aperiodic points in $X(\sigma)$.
\end{corollary}

In the same way, we have a simple proof of \cite[Proposition~5.1]{BealPerrinRestivo2023} concerning the representability of~$\sigma$. 

\begin{corollary} \label{c:representable}
For any morphism $\sigma$, any point in $X(\sigma)$ has a $\sigma$-representation in $X(\sigma)$.
\end{corollary}

Finally, we consider the tightness of the bounds in Proposition~\ref{p:main} and Theorem~\ref{theoremBSTY}.
We have already seen in Example~\ref{ex:nonrepresentable} that a sequence can be non-representable at level $0$ with $\Card(A_1) = 2$.
An example of a primitive sequence of morphisms with $\Card(A_n) = 3$ for all $n \ge 1$ that is not recognizable for aperiodic points at level~$0$ is given in \cite[Example~4.3]{BertheSteinerThuswaldnerYassawi2019}.
The following example shows that we can have $r{-}2$ levels of
non-recognizability for aperiodic points, where $r$ is the finite alphabet rank.

\begin{example} \label{ex:Knonrec}
Let $K \ge 1$, $A_n = \{a_0,a_1,\dots,a_n\}$ for $0 \le n \le K$, $A_n = \{a_0,a_1,\dots,a_{K+1}\}$ for $n > K$, and 
\begin{align*}
\sigma_n & \colon a_i \mapsto a_0 a_i a_0 \text{ for all } 0 \le i \le n,\ a_{n+1} \mapsto a_n, & & \text{for } 0 \le n \le K, \\
\sigma_n & \colon a_i \mapsto a_0 a_i a_0 \text{ for all } 0 \le i \le K+1, & & \text{for } n > K.
\end{align*}
Then, for all $0 \le n \le K+1$, $X^{(n)}(\sigma)$ consists of the closure
of the shift orbits of the sequences $\cdots a_0a_0\cdot a_ia_0a_0 \cdots$, $0 \le i \le n$. 
For $0 \le n \le K$, we have 
\begin{displaymath}
\sigma_n(\cdots a_0a_0\cdot a_na_0a_0 \cdots) = \cdots a_0\cdot a_0a_na_0a_0 \cdots = \sigma_n(\cdots a_0\cdot a_0a_{n+1}a_0a_0 \cdots),
\end{displaymath}
thus $\sigma$ is not recognizable at level $n$ for aperiodic points for all $1 \le n \le K$ (and not recognizable at level~$0$).
\end{example}

Example~\ref{ex:Knonrec} can be easily modified to obtain a sequence of morphisms that is not recognizable at any level $n$ for aperiodic points.

\begin{example}
Let $A_n = \{a_0,a_1,\dots,a_{n+1}\}$ for all $n \ge 0$, and
\begin{displaymath}
\sigma_n \colon a_i \mapsto a_0 a_i a_0 \text{ for all } 0 \le i \le n{+}1,\ a_{n+2} \mapsto a_{n+1}, \quad \text{for all } n\ge 0.
\end{displaymath}
Then $\sigma$ is not recognizable at level $n$ for aperiodic points for all $n \ge 0$ because
\begin{displaymath}
\sigma_n(\cdots a_0a_0\cdot a_{n+1}a_0a_0 \cdots) = \cdots a_0\cdot a_0a_{n+1}a_0a_0 \cdots = \sigma_n(\cdots a_0\cdot a_0a_{n+2}a_0a_0 \cdots).
\end{displaymath}
\end{example}

We do not know whether the bound of $r{-}1$ levels of
non-representability is tight, where $r$ is the finite alphabet rank.

\smallskip
A final remark concerns the existence of a recognizable $S$-adic representation for a given shift space~$X$. 
Can one always modify the morphisms of a sequence~$\sigma$ in such a way that $X = X^{(0)}(\sigma')$ with $\sigma'$ recognizable (keeping properties of $\sigma$
like primitivity, constant length, etc..)?
The answer is known in the particular case of automatic shifts, which are $S$-adic shifts defined by sequences $(\varphi,\sigma,\sigma,\ldots)$, where $\sigma$ is a morphism of constant length and $\varphi$ a letter coding (that is, has constant length~$1$).
It has been shown in \cite{MuellnerYassawi2021} that, for aperiodic automatic shifts generated by a primitive morphism~$\sigma$, we can assume w.l.o.g.\ that $\varphi$ is injective on $X(\sigma)$, thus $(\varphi,\sigma,\sigma,\ldots)$ is recognizable. 
It is not known if a similar result holds for general morphic shifts, i.e., when $\sigma$ is not of constant length. 
Other results exhibiting recognizable sequences of morphisms can be found in~\cite{Espinoza2023}.


\section*{Acknowledgments}
This work was supported by the Agence Nationale
de la Recherche (ANR-18-CE40-0007 and ANR-22-CE40-0011).

\bibliographystyle{plain}
\bibliography{recognizability}

\begin{thebibliography}{10}

\bibitem{BealPerrinRestivo2023}
Marie{-}Pierre B{\'{e}}al, Dominique Perrin, and Antonio Restivo.
\newblock Recognizability of morphisms.
\newblock {\em Erg. Theor. {\&} Dyn. Sys.}, 43(11):3578--3602, 2023.

\bibitem{BertheSteinerThuswaldnerYassawi2019}
Val\'{e}rie Berth\'{e}, Wolfgang Steiner, J\"{o}rg~M. Thuswaldner, and Reem
  Yassawi.
\newblock Recognizability for sequences of morphisms.
\newblock {\em Erg. Theor. {\&} Dyn. Sys.}, 39(11):2896--2931, 2019.

\bibitem{BezuglyiKwiatkowskiMedynets2009}
Sergey Bezuglyi, Jan Kwiatkowski, and Konstantin Medynets.
\newblock Aperiodic substitution systems and their {B}ratteli diagrams.
\newblock {\em Erg. Theor. {\&} Dyn. Sys.}, 29(1):37--72, 2009.

\bibitem{BustosGajardoManiboYassawi2023}
\'{A}lvaro Bustos-Gajardo, Neil Ma\~{n}ibo, and Reem Yassawi.
\newblock Torsion-free {$S$}-adic shifts and their spectrum.
\newblock {\em Studia Math.}, 272(2):159--198, 2023.

\bibitem{DurandPerrin2021}
Fabien Durand and Dominique Perrin.
\newblock {\em Dimension groups and dynamical systems}.
\newblock Cambridge University Press, 2021.

\bibitem{EhrenfeuchtRozenberg1978}
Andrew Ehrenfeucht and Gregorz Rozenberg.
\newblock Elementary homomorphisms and a solution of the {${\rm DOL}$} sequence
  equivalence problem.
\newblock {\em Theoret. Comput. Sci.}, 7(2):169--183, 1978.

\bibitem{Espinoza2023}
Basti\'{a}n Espinoza.
\newblock Symbolic factors of {$\mathcal{S}$}-adic subshifts of finite alphabet
  rank.
\newblock {\em Erg. Theor. {\&} Dyn. Sys.}, 43(5):1511--1547, 2023.

\bibitem{KarhumakiManuchPlandowski2003}
Juhani Karhum\"{a}ki, J\'{a}n Ma\v{n}uch, and Wojciech Plandowski.
\newblock A defect theorem for bi-infinite words.
\newblock {\em Theoretical Computer Science}, 292(1):237--243, 2003.
\newblock Selected Papers in honor of Jean Berstel.

\bibitem{Kyriakoglou2019}
Revekka Kyriakoglou.
\newblock {\em Recognizable substitutions}.
\newblock PhD thesis, Universit\'e Paris Est, 2019.

\bibitem{LindMarcus2021}
Douglas Lind and Brian Marcus.
\newblock {\em An Introduction to Symbolic Dynamics and Coding}.
\newblock Cambridge Mathematical Library. Cambridge University Press, 2nd
  edition, 2021.

\bibitem{Martin1973}
John~C. Martin.
\newblock Minimal flows arising from substitutions of non-constant length.
\newblock {\em Math. Systems Theory}, 7:72--82, 1973.

\bibitem{Mosse1992}
Brigitte Moss{\'e}.
\newblock Puissances de mots et reconnaissabilit\'e des points fixes d'une
  substitution.
\newblock {\em Theoret. Comput. Sci.}, 99(2):327--334, 1992.

\bibitem{Mosse1996}
Brigitte Moss\'{e}.
\newblock Reconnaissabilit\'{e} des substitutions et complexit\'{e} des suites
  automatiques.
\newblock {\em Bull. Soc. Math. France}, 124(2):329--346, 1996.

\bibitem{MuellnerYassawi2021}
Clemens {M\"ullner} and Reem {Yassawi}.
\newblock Automorphisms of automatic shifts.
\newblock {\em Erg. Theor. {\&} Dyn. Sys.}, 41(5):1530--1559, 2021.

\end{thebibliography}
\end{document}